\documentclass[a4paper, 11pt]{article}

\usepackage{amsmath, amssymb, amsthm}
\usepackage[english]{babel}
\usepackage[T1]{fontenc}

\newcommand{\be}{\begin{equation}}
\newcommand{\ee}{\end{equation}}
\newcommand{\bd}{\begin{displaymath}}
\newcommand{\ed}{\end{displaymath}}
\newcommand{\ba}{\begin{eqnarray}}
\newcommand{\ea}{\end{eqnarray}}

\def\R{{I \!\! R}}

\def\mD{ {\mathcal {D}}}

\def\e{\epsilon}

\def\v12{(v-w)}

\def\({\left(}
\def\){\right)}

\def\bgr#1\egr{{\allowdisplaybreaks\begin{gather}#1\end{gather}}}
\def\bma#1\ema{{\allowdisplaybreaks\begin{align}#1\end{align}}}

\def\oplem#1{\begin{lemma}\, {\rm #1}\, \it }
\def\cllem{\end{lemma}\rm \par }
\def\opthm#1{\begin{theorem}\, {\rm #1}\, \it }
\def\clthm{\end{theorem}\rm \par }

\def\N{\mathbb{N}}

\def\N{\mathbb{N}}
\def\R{\mathbb{R}}

\newcommand{\fer}[1]{(\ref{#1})}
\newcommand{\bq}{\begin{equation}}
\newcommand{\eq}{\end{equation}}
\def\bqa{\begin{eqnarray}}
\def\eqa{\end{eqnarray}}
\def\bd{\begin{displaymath}}
\def\ed{\end{displaymath}}

\newtheorem{rmk}{Remark}

\renewcommand{\(}{\left(}
\renewcommand{\)}{\right)}

\usepackage{color}

\newtheorem{thm}{Theorem}

\newtheorem{lem}[thm]{Lemma}

\theoremstyle{remark}

\theoremstyle{definition}

\newenvironment{equations}{\equation\aligned}{\endaligned\endequation}

\usepackage{amssymb}

\begin{document}

\title{The fractional Fisher information and the \\ central limit theorem for stable laws}

 \author{Giuseppe Toscani \thanks{Department of Mathematics, University of Pavia, via Ferrata 1, 27100 Pavia, Italy.
\texttt{giuseppe.toscani@unipv.it} }}

\maketitle

\begin{center}\small
\parbox{0.85\textwidth}{
\textbf{Abstract.}
A new information-theoretic approach to the central limit theorem for stable laws is
presented. The main novelty is the concept of relative fractional Fisher information,
which shares most of the properties of the classical one, included  Blachman-Stam type
inequalities.  These inequalities relate the fractional Fisher information of the sum
of $n$ independent random variables to the information contained in sums over subsets
containing $n-1$ of the random variables.  As a consequence, a simple proof of the
monotonicity of the relative fractional Fisher information in central limit theorems
for stable law is obtained, together with an explicit decay rate.
}
\end{center}

\medskip



\section{Introduction}

The entropy functional (or Shannon's entropy) of a real valued random variable $X$
with density $f$ is defined as
 \be\label{Shan}
H(X) = H(f) = - \int_{\R} f(x) \log f(x)\, dx.
 \ee
provided that the integral makes sense. Among random variables with the same variance
$\sigma$ the standard Gaussian $Z$ with variance $\sigma$ has the largest entropy. If
the $X_j$'s are independent copies of a centered random variable $X$ with variance
$1$, then the (classical) central limit theorem implies that the law of the normalized
sums
 \[
S_n = \frac 1{\sqrt n}\sum_{j=1}^n X_j
 \]
converges weakly to the law of the centered standard Gaussian $Z$, as $n$ tends to infinity.

A direct consequence of the entropy power inequality, postulated by Shannon \cite{Sha}
in the fourthies, and subsequently proven by Stam  \cite{Sta} (cf. also Blachman
\cite{Bla}), implies that $H(S_2)\ge H(S_1)$. The entropy of the normalized sum of two
independent copies of a random variable is larger than that of the original. A shorter
proof was obtained later by Lieb \cite{Lieb} (cf. also  \cite{Bar, Joh, JB} for
exhaustive presentation of the subject). While inductively expected that  the entire
sequence $H(S_n)$ should increase with $n$, as conjectured  by Lieb in $1978$
\cite{Lieb}, a rigorous  proof of this result was found only $25$ years later by
Artstein, Ball,  Barthe and A. Naor  \cite{ABBN1, ABBN2}.

More recently, simpler proofs of the monotonicity of the sequence $H(S_n)$ have been
obtained by Madiman and Barron \cite{MB, BM} and Tulino and Verd\'u \cite{TV}. Madiman and
Barron \cite{MB}, by means of a detailed analysis of variance projection properties,
derived new entropy power inequalities for sums of independent random variables, and,
as a consequence, the monotonicity of entropy in central limit theorems for
independent and identically distributed random variables. Tulino and Verd\'u \cite{TV}
obtained analogous results by taking advantage of projection through minimum
mean-squared error interpretation. As observed in \cite{BM}, the proofs of the main
result in both \cite{BM, TV} share essential similarities.

As suggested by Stam's proof of the entropy power inequality \cite{Bla, Sta}, most of the results about monotonicity  benefit from the reduction from entropy to another information-theoretic notion, the Fisher information of a random variable. For sufficiently regular densities, the Fisher information can
be written as
 \be\label{fish}
I(X) = I(f) = \int_{\{f>0\}} \frac{|f'(x)|^2}{f(x)} \, dx.
 \ee
Among random variables with the same variance $\sigma$, the Gaussian $Z$ has smallest Fisher
information $1/\sigma$. Fisher information and entropy are related each other by the so-called
de Bruijn relation \cite{Bar, Sta}. If  $u(x,t)= u_t(x)$ denotes the solution to the initial value problem for the heat equation in the whole space $\R$,
 \be\label{heat}
 \frac{\partial u}{\partial t} =  \frac{\partial^2 u}{\partial x^2},
 \ee
 leaving from an initial probability density function $f(x)$,
 \[
 I(f) = \frac d{dt} H(u_t)|_{t=0}.
 \]
A particularly clear explanation of this link is given in the article of Carlen and Soffer \cite{CS} (cf. also Barron \cite{Bar} and Brown \cite{Bro}).

It is noticeable that the connection between Fisher information and the central limit
theorem was noticed at the time of Stam's proof of entropy power inequality by Linnik
\cite{Lin}, who first used Fisher information in a proof of the central limit theorem.

Recently, the role of Fisher information in limit theorems has been considered also in
situations different from the classical central limit theorem. In \cite{BCG3} (cf.
also \cite{BCG1, BCG2}) Bobkov,  Chistyakov  and  G\"otze enlightened the possibility
to make use of the relative Fisher information to study convergence towards a stable
law \cite{Fel, GK, LR79}. If the $X_j$'s are independent copies of a centered random
variable $X$ which lies in the domain of normal attraction of a random variable
$Z_\lambda$ with {\it L\'evy symmetric stable distribution}, the central limit theorem
for stable laws implies that the law of the normalized sums
 \be\label{stab}
T_n = \frac 1{n^{1/\lambda}}\sum_{j=1}^n X_j
 \ee
converges weakly to the law of the centered stable $Z_\lambda$, as $n$ tends to
infinity. Given a random variable $X$, with probability distribution $G(x)$, the
Fourier transform of $G(x)$ will be denoted by
        \[
       {\cal F}G(\xi)= \hat G(\xi) =\int_\R  \e^{-i\xi x}\, d G(x),\quad  \xi\in\R.
        \]
In case $X$ possesses a probability density $g(x)= G'(x)$, we will still denote the
Fourier transform of $g(x)$ by
  \[
       {\cal F}g(\xi)= \hat g(\xi) =\int_\R  \e^{-i\xi x}\, g(x) \, dx,\quad  \xi\in\R.
        \]

A {L\'evy symmetric stable distribution} $L_\lambda(x)$ of order $\lambda$ is defined
in terms of its Fourier transform by
 \be\label{levy}
 \widehat L_\lambda(\xi) = \e^{-|\xi|^\lambda}.
 \ee
While the Gaussian density is related to the linear diffusion equation \fer{heat},  L\'evy distributions are deeply related to linear fractional diffusion equations
 \be\label{frac1}
 \frac{\partial u(x,t)}{\partial t}= \mathcal D_{2\alpha} u(x,t).
 \ee
For the classical diffusion case described by \fer{heat}, $\alpha=1$ and the diffusion
operator models a Brownian diffusion process. For fractional diffusion, where    $1/2
< \alpha < 1$,  the $\mathcal D_{2\alpha}$ operator in \fer{frac1} is commonly
referred to as \emph{anomalous diffusion}, and the underlying stochastic process is a
L\'evy stable flight.

Indeed, the fractional diffusion equation \fer{frac1} can be fruitfully described in terms of Fourier variables, where it takes the form
 \be\label{frac-fou}
 \frac{\partial \hat u(\xi,t)}{\partial t}=-|\xi|^{2\alpha} \hat u(\xi,t).
 \ee
Equation \fer{frac-fou} can be easily solved. Its solution
 \be\label{fund}
\hat u(\xi,t) = \hat u_0(\xi)\e^{-|\xi|^{2\alpha} t}.
 \ee
shows that the fractional diffusion equation admits a fundamental solution, given by a
scaled in time L\'evy distribution of order $\lambda= 2\alpha$.

Equation \fer{fund} enhances a strong analogy between the solution of the heat
equation \fer{heat} and the solution to
the fractional diffusion equation \fer{frac1}, and suggests that information
techniques used for the former could be fruitfully used, by suitably adapting these
techniques to the new situation, to the latter \cite{FPTT}.

At difference with the analysis of Bobkov,  Chistyakov  and  G\"otze \cite{BCG3}, in
this note we will show that this analogy can be reasonably stated at the level of
Fisher information by resorting to a new definition, which in our opinion is better
adapted to the anomalous diffusion case. In Section \ref{score}, we introduce a
generalization of (relative) Fisher information, the relative fractional Fisher
information, that is constructed to vanish on the L\'evy symmetric stable
distribution, and shares most of the properties of the classical one defined by
\fer{fish}. By using this new definition, we will present in Section \ref{mono} an
information-theoretic proof of monotonicity of the fractional Fisher information of
the normalized sums \fer{stab} by adapting the ideas developed in \cite{BM}. At
difference with the classical case, it will be showed in Section \ref{mono} that the
monotonicity also implies convergence in fractional Fisher information at a precise
rate, which depends on the value of the exponent that characterizes the L\'evy
distribution.

Fractional in space diffusion equations appear in many contexts. Among others, the
review paper by Klafter et al. \cite{KZS97} provides numerous references to physical
phenomena in which these anomalous diffusions occur (cf. \cite{BWM00, Cha98, GM98,
MFL02, SBMW01} and the reference therein for various details on both mathematical and
physical aspects). Also, fractional diffusion equations in the nonlinear setting have
been intensively studied in the last years by Caffarelli, Vazquez et al. \cite{CV,
Car, Va1, Va2}. The reading of these papers was essential in moving my interest
towards the present topic.

\section{Scores and fractional Fisher information}\label{score}

In the rest of this paper, if not explicitly quoted, and without loss of generality,
we will always assume that any random variable $X$ we will consider is centered, i.e.
$E(X)=0$, where as usual $E(\cdot)$ denotes mathematical expectation.

In this section we briefly summarize the mathematical notations and the  meaning of
the fractional diffusion operator $\mathcal D_{2\alpha}$ which appears in equation
\fer{frac1}. For $0 <\alpha < 1$ we let $R_\alpha$ be the one-dimensional
\emph{normalized} Riesz potential operator defined for locally integrable functions by
\cite{Rie, Ste}
 \[
 R_\alpha(f)(x) = S(\alpha) \int_\R \frac{f(y)\, dy}{|x-y|^{1-\alpha}}.
 \]
The constant $S(\alpha)$ is chosen to have
 \be\label{rt}
\widehat{R_\alpha(f)}(\xi) = |\xi|^\alpha \widehat f(\xi).
 \ee
Since for $0 <\alpha < 1$ it holds  \cite{Lieb83}
 \be\label{hom}
 \mathcal F |x|^{\alpha -1} = |\xi|^{-\alpha} \pi^{1/2} \Gamma \left(\frac{1-\alpha}2 \right) \Gamma \left(\frac{\alpha}2 \right),
 \ee
where, as usual $\Gamma(\cdot)$ denotes the Gamma function, the value of $S(\alpha)$
is given by
 \[
 S(\alpha) = \left[\pi^{1/2} \Gamma \left(\frac{1-\alpha}2 \right) \Gamma \left(\frac{\alpha}2 \right)\right]^{-1}.
 \]
 Note that $S(\alpha) = S(1-\alpha)$.

We then define the fractional derivative of order $\alpha$ of a real function $f$ as
($0 <\alpha < 1$)
 \be\label{fa}
 \frac{d^\alpha f(x)}{dx^\alpha} = \mD_\alpha f(x) = \frac{d}{dx}R_{1-\alpha}(f)(x).
 \ee
Thanks to \fer{rt}, in Fourier variables
 \be\label{d1}
 \widehat{\mD}_\alpha f(\xi) = i \frac{\xi}{|\xi|} |\xi|^\alpha \widehat{f}(\xi).
 \ee

In theoretical statistics, the score or efficient score \cite{CH, BM} is the
derivative, with respect to some parameter $\theta$, of the logarithm of the
likelihood function (the log-likelihood). If the observation is $X$ and its likelihood
is $L(\theta;X)$, then the score $\rho_L(X)$ can be found through the chain rule
 \be
\rho_L(\theta, X) = \frac{1}{L(\theta;X)} \frac{\partial L(\theta;X)}{\partial\theta}.
\ee Thus the score indicates the sensitivity of $L(\theta;X)$ (its derivative
normalized by its value). In older literature, the term \emph{linear score} refers to
the score with respect to an infinitesimal translation of a given density. In this
case, the likelihood of an observation is given by a density of the form
$L(\theta;X)=f(X+\theta)$. According to this definition, given a random variable $X$
in $\R$ distributed with a differentiable probability density function $f(x)$, its
linear score $\rho$ (at $\theta=0$) is given by
 \be\label{sco1}
\rho(X) = \frac{f'(X)}{f(X)}.
 \ee
The linear score has zero mean, and its variance is just the Fisher information
\fer{fish} of $X$.

Also, the notion of relative score has been recently considered in information theory
\cite{Guo} (cf. also \cite{BCG3}).  For every pair of random variables $X$ and $Y$
with differentiable density functions $f$ (respectively $g$), the score function of
the pair relative to $X$ is represented by
 \be\label{rel-sco}
\tilde\rho(X) = \frac{f'(X)}{f(X)} - \frac{g'(X)}{g(X)}.
 \ee
In this case, the relative to $X$ Fisher information between $X$ and $Y$ is just the
variance of $\tilde\rho(X)$. This notion  is satisfying because it represents the
variance of some error due to the mismatch between the prior distribution $f$ supplied
to the estimator and the actual distribution $g$. Obviously, whenever $f$ and $g$ are
identical, then the relative Fisher information is equal to zero.

Let $z_\sigma(x)$ denote the Gaussian density in $\R$ with zero mean and variance
$\sigma$
 \be\label{max}
z_\sigma(x) = \frac 1{\sqrt{2\pi \sigma}}\exp\left(- \frac{|x|^2}{2\sigma}\right).
 \ee
Then a Gaussian random variable of density $z_\sigma$ is uniquely defined by the score
function
 \[
 \rho(Z_\sigma) = - Z_\sigma/\sigma.
 \]
Also, the relative (to $X$) score function of $X$ and $Z_\sigma$ takes the simple
expression
 \be\label{rel-z}
\tilde\rho(X) = \frac{f'(X)}{f(X)} + \frac X\sigma,
 \ee
which induces a (relative to the Gaussian) Fisher information
 \be\label{fish-r}
\tilde I(X) = \tilde I(f) = \int_{\{f>0\}} \left( \frac{f'(x)}{f(x)} + \frac x\sigma
\right)^2f(x)\, dx.
 \ee
Clearly, $\tilde I(X) \ge 0$, while $\tilde I(X) = 0$ if $X$ is a centered Gaussian
variable of variance $\sigma$.

 The concept of linear score can be naturally extended
to cover fractional derivatives. Given a random variable $X$ in $\R$ distributed with
a probability density  function $f(x)$ that has  a well-defined  fractional derivative
of order $\alpha$, with $0<\alpha < 1$, its linear fractional score, denoted by
$\rho_{\alpha+1}$ is given by
 \be\label{sco2}
\rho_{\alpha+1}(X) = \frac{\mD_\alpha f(X)}{f(X)}.
 \ee
Thus the linear fractional score indicates the non local (fractional) sensitivity of
$f(X+\theta)$ at $\theta =0$ (its fractional derivative normalized by its value).
Differently from the classical case, the fractional score of $X$ is linear in $X$ if
and only if $X$ is a L\'evy distribution of order $\alpha +1$. Indeed, for a given
positive constant $C$, the identity
 \[
\rho_{\alpha+1}(X)  = - C X,
 \]
is verified  if and only if, on the set $\{f>0\}$
 \be\label{ccw}
\mD_\alpha f(x) = -Cxf(x).
 \ee
Passing to Fourier transform, this identity yields
 \[
i \xi |\xi|^{\alpha-1} \widehat f(\xi) = -i C \frac{\partial \widehat f(\xi)}{\partial
\xi},
 \]
and from this follows
 \be\label{ok3}
\widehat f(\xi) = \widehat f(0) \exp\left\{-\frac{|\xi|^{\alpha +1}}{C(\alpha
+1)}\right\}.
 \ee
Finally,  by choosing $C = (\alpha+1)^{-1}$, and imposing that $f(x)$ is a probability
density function (i.e. by fixing $\widehat f(\xi= 0) =1$), we obtain  that the L\'evy
stable law of order $\alpha + 1$ is the unique probability density solving  \fer{ccw}.

It is important to remark that, unlike in the case of the linear score, the variance of the fractional score is in general
unbounded. One can easily realize this by looking at the variance of the
fractional score in the case of a L\'evy variable. For a L\'evy variable, in fact, the
variance of the fractional score coincides with a multiple of its variance, which is
unbounded \cite{GK,LR79}. For this reason, a consistent definition in this case is
represented by the relative fractional score. In reason of \fer{ok3}, a L\'evy random
variable of density $z_\lambda$, with $1<\lambda<2$ is uniquely defined by a linear
fractional score function
 \[
 \rho_\lambda (Z_{\lambda}) = -\frac{Z_{\lambda}}{\lambda},
 \]
the relative (to $X$) fractional score function of $X$ and $Z_{\lambda}$ assumes the
simple expression
 \be\label{rel-z2}
\tilde\rho_\lambda(X) = \frac{\mD_{\lambda-1} f(X)}{f(X)} + \frac X{\lambda},
 \ee
which induces a (relative to the L\'evy) fractional Fisher information (in short
$\lambda$-Fisher relative information)
 \be\label{fish-r2}
I_\lambda(X) =  I_\lambda(f) = \int_{\{f>0\}} \left( \frac{\mD_{\lambda-1} f(x)}{f(x)} +
\frac x\lambda \right)^2f(x)\, dx.
 \ee
The fractional Fisher information is always greater or equal than zero, and it is
equal to zero if and only if $X$ is a L\'evy symmetric stable distribution of order
$\lambda$. At difference with the relative standard relative Fisher information,
$I_\lambda$ is well-defined any time that the the random variable $X$ has a
probability density function which is suitably closed to the L\'evy stable law
(typically lies in a subset of the domain of attraction). We will define by $\mathcal
P_\lambda$ the set of probability density functions such that $I_\lambda(f) <
+\infty$, and we will say that a random variable $X$ lies in the domain of attraction
of the $\lambda$-Fisher information if $I_\lambda(X) < +\infty$. More in general, for
a given positive constant $\upsilon$, we will consider other relative fractional score
functions given by
 \be\label{rel-z3}
\tilde\rho_{\lambda, \upsilon}(X) = \frac{\mD_{\lambda-1} f(X)}{f(X)} + \frac
X{\lambda\upsilon}.
 \ee
This leads to the relative fractional Fisher information
 \be\label{fish-r3}
 I_{\lambda,\upsilon}(X) =  I_{\lambda,\upsilon}(f) = \int_{\{f>0\}} \left( \frac{\mD_{\lambda-1}
f(x)}{f(x)} + \frac x{\lambda\upsilon} \right)^2f(x)\, dx.
 \ee
Clearly,  $I_\lambda= I_{\lambda, 1}$. Analogously, we will define by $\mathcal
P_{\lambda,\upsilon}$ the set of probability density functions such that
$I_{\lambda,\upsilon}(f) < +\infty$, and we will say that a random variable $X$ lies
in the domain of attraction  if $I_{\lambda, \upsilon}(X) < +\infty$.

\begin{rmk}
{\rm The characterization of the functions which belong to the domain of attraction of
the relative fractional Fisher information is not an obvious task. However, it can be
seen that this set of functions is not empty. We will present an explicit example in
the Appendix.}
\end{rmk}

\section{Monotonicity of the fractional Fisher information}\label{mono}

We now proceed to recover some useful properties of the relative fractional score
function $\tilde\rho_\lambda(X)$ defined in \fer{rel-z2}. Most of these properties are
easy generalizations of analogous ones proven in \cite{BM} for the standard linear
score function. The main difference here is that we need to resort to the relative
one. The following Lemma holds

\begin{lem}\label{l1}
Let $X_1$ and $X_2$ be independent random variables with smooth densities, and
let $\rho^{(1)}$ (respectively $\rho^{(2)}$) denote their fractional scores.
Then, for each constant $\lambda$, with $1<\lambda <2$, and each positive constant $\delta$, with $0< \delta < 1$, the relative fractional score function of the sum $X_1 + X_2$ can be expressed as
 \begin{equation}\label{sum}
 \tilde\rho_\lambda(x) =
 E\left[ \delta \, \tilde\rho_{\lambda,\delta}^{(1)}(X_1)
 +(1- \delta)\, \tilde\rho_{\lambda,1-\delta}^{(2)}(X_2) \big| X_1+X_2
 =x\right].
 \end{equation}
 \end{lem}

\begin{proof}
Let $f_j$, $j =1,2$ and $f$ be the densities of $X_j$, $j=1,2$ and $X_1 +X_2$. Then,
the density of $X_1+X_2$ is given by the convolution product of $f_1$ and $f_2$,  $f=
f_1*f_2$. To start with, we remark that the fractional derivatives, as given by
\fer{fa} and \fer{d1}, share the same behavior, with respect to convolutions, of the
usual derivatives. Indeed, thanks to \fer{d1}, it follows that
 \[
 \mD_{\lambda-1}f(x) = \left(\mD_{\lambda-1} f_1\right)*f_2(x)= f_1*\left(\mD_{\lambda-1} f_2\right)(x).
 \]
The previous identity can be rewritten in terms of expectations as \cite{BM}
 \begin{equations}\nonumber
 \mD_{\lambda-1} f(x) =& \mD_{\lambda-1} E[f_1(x-X_2)]\\
 =& E[\mD_{\lambda-1} f_1(x-X_2)] \\
 =&  E[\rho_\lambda^{(1)}(x-X_2)f_1(x-X_2)].
 \end{equations}
Therefore
  \begin{equations}\nonumber
 \rho_\lambda(x) = &\frac{\mD_{\lambda-1} f(x)}{f(x)}\\
  =&  E\left[\rho_\lambda^{(1)}(x-X_2)\frac{f_1(x-X_2)}{f(x)}\right] \\
  =& E\left[\rho_\lambda^{(1)}(X_1)\big| X_1 + X_2=x\right].
 \end{equations}
As usual, given the random variables $X$ and $Y$, we denoted by $E[X|Y]$ the
conditional expectation of $X$ given $Y$. Exchanging the indexes, we obtain an
identical expression which relates the fractional score of the sum to the second
variable $X_2$. Hence, for each positive constant $\delta$, with $0<\delta <1$ we have
 \be\label{im}
\rho_\lambda(x) = \delta E\left[\rho_\lambda^{(1)}(X_1)\big| X_1 + X_2=x\right] +
(1-\delta) E\left[\rho_\lambda^{(2)}(X_2)\big| X_1 + X_2=x\right].
 \ee
To conclude the proof, it is enough to remark that the following identity holds true
 \[
 x = \delta E\left[\frac{X_1}\delta\big| X_1 + X_2=x\right] +
(1-\delta) E\left[\frac{X_2}{1-\delta}\big| X_1 + X_2=x\right].
 \]
\end{proof}

Lemma \ref{l1} ha several interesting consequences.  Since the norm of the relative
fractional score is not less than that of its projection (i.e. by the Cauchy--Schwarz
inequality), we obtain
 \begin{equations}\label{BS}
 I_\lambda (X_1+X_2) =& E\left[ \tilde\rho_\lambda^2(X_1+X_2)\right] \le
  E\left[\left(  \delta \, \tilde\rho_{\lambda,\delta}^{(1)}(X_1)
 +(1- \delta)\, \tilde\rho_{\lambda,1-\delta}^{(2)}(X_2) \right)^2\right] = \\
 & \delta^2  I_{\lambda, \delta} (X_1) + (1-\delta)^2I_{\lambda,1- \delta} (X_2).
\end{equations}
Inequality \fer{BS} is the analogous of the Blachman--Stam inequality \cite{Bla, Sta},
and allows to bound the relative fractional Fisher information of the sum of
independent variables   in terms of the relative fractional Fisher information of its
addends. Inequality \fer{BS} can be reduced to a normal form by making use of scaling
arguments. Indeed, for any given random variable $X$ such that one of the two sides is
bounded, and positive constant $\upsilon$, the following identity holds
 \be\label{scal}
 I_{\lambda,\upsilon}(\upsilon^{1/\lambda}X) = \upsilon^{-2(1-1/\lambda)} I_\lambda
 \left(X\right).
 \ee
Using identity \fer{scal} into inequality \fer{BS} we can write it in the form
 \be\label{BS1}
I_\lambda (\delta^{1/\lambda}X_1+(1-\delta)^{1/\lambda}X_2)\le \delta^{2/\lambda}
I_\lambda\left(X_1\right) + (1- \delta)^{2/\lambda} I_\lambda\left( X_2\right).
 \ee
Note that equality into \fer{BS1} holds if and only if both $X_1$ and $X_2$ are L\'evy
variables of order $\lambda$. Indeed, equality into \fer{BS} holds if and only if the
relative score functions satisfy
 \be\label{33}
\tilde\rho_{\lambda}^{(1)}(x) =c_1; \quad \tilde\rho_{\lambda}^{(2)}(x)=c_2.
 \ee
In fact, if this is the case, the Cauchy--Schwarz inequality we used to obtain
\fer{BS} is satisfied with the equality sign. On the other hand, proceeding as in the
derivation of \fer{ok3}, we conclude that \fer{33} implies that the densities of the
$X_j$'s, $j=1,2$ have Fourier transforms
 \be\label{ok2}
\widehat f_j(\xi) = \exp\left\{-|\xi|^{\lambda} + ic_j\xi \right\}.
 \ee
 We proved
\begin{thm}\label{bl} Let $X_j$, $j_1,2$ be independent random variables such that their relative
fractional Fisher information  functions $I_\lambda(X_j)$, $j=1,2$ are bounded for
some $\lambda$, with $1<\lambda <2$. Then, for each constant $\delta$ with $0<\delta <
1$,  $I_\lambda(\delta^{1/\lambda}X_1+(1-\delta)^{1/\lambda}X_2) $ is bounded, and
inequality \fer{BS1} holds. Moreover, there is equality in \fer{BS1} if and only if,
up to translation, both $X_j$, $j=1,2$ are L\'evy variables of exponent $\lambda$.
\end{thm}

\emph{A posteriori}, we can use the result of Theorem \ref{bl} to avoid inessential
difficulties in proofs by means of a smoothing argument. Indeed, since we are
interested in inequalities for convolutions of densities, and a L\'evy symmetric
stable law is stable with respect to the operation of convolution, we can consider in
the following text densities suitably smoothed by convolution with a L\'evy symmetric
stable law. In fact, if $Z_1$ and $Z_2$ denote two independent copies of a symmetric
L\'evy stable law $Z$ of order $\lambda$, for any given positive constants
$\epsilon_j$, $j=1,2$ the random variable $\epsilon_1^{1/\lambda}Z_1 +
\epsilon_2^{1/\lambda}Z_2$ is  symmetric L\'evy stable with the law of
$(\epsilon_1+\epsilon_2)^{1/\lambda}Z$. Therefore, for any given positive constant
$\epsilon <1$, and random variable $X$ with density function $f$ we will denote by
$f_\epsilon$ the density of the random variable $X_\epsilon$ given by
$(1-\epsilon)^{1/\lambda}X + \epsilon^{1/\lambda}Z$, where the symmetric stable L\'evy
variable $Z$ of order $\lambda$ is independent of $X$.

By virtue of Theorem \ref{bl},
 \be\label{ep}
 I_{\lambda}(X_\epsilon) \le (1-\epsilon)^{2/\lambda} I_\lambda \left(X\right).
 \ee

Hence, the relative fractional Fisher information of the smoothed version $X_\epsilon$
of the random variable  is always smaller than the relative fractional Fisher
information of $X$. Moreover,
 \[
 \lim_{\epsilon \to 0} I_{\lambda}(X_\epsilon) = I_\lambda \left(X\right).
 \]
This statement follows from \fer{ep} and Fatou's lemma. Indeed, in view of \fer{ep},
it only remains to show that $\liminf_{\epsilon \to 0}I_{\lambda}(X_\epsilon) \ge
I_\lambda \left(X\right)$. On $\{f>0\}$,  $(\mathcal D_{\lambda
-1}f_\epsilon)^2/f_\epsilon$ converges a.e. to $(\mathcal D_{\lambda -1}f)^2/f$. This
is enough to imply by Fatou's lemma the desired inequality.

The next ingredient in the proof of monotonicity deals with the so-called variance
drop inequality \cite{BM}. The idea goes back at least to the pioneering work of
Hoeffding on $U$ statistics \cite{Hoe}. Let $[n]$ denote the index set
$\{1,2,\dots,n\}$, and, for any ${\bf s} \subset [n]$, let $X_{\bf s}$ stand for the
collection of random variables $(X_i: i \in {\bf s})$, with the indices taken in their
natural increasing order. Then we have

\begin{thm}\label{kk} Let the function $\Phi : \R^m \to \R$, with $1 \le m \in \N$,
be symmetric in its
arguments, and suppose that $E\left[\Phi(X_1,X_2, \dots, X_m) \right]=0$. Define
 \be\label{good}
U(X_1,X_2,\dots, X_n) = \frac{m!(n-m)!}{n!} \sum_{\left\{ {\bf s} \subset [n]: |{\bf
s}| = m\right\}} \Phi\left( X_{\bf s}\right).
  \ee
Then
 \be\label{hof}
 E\left[U^2\right] \le \frac mn E\left[\Phi^2\right].
 \ee
\end{thm}
In theoretical statistical, $U$ defined in \fer{good} is called a $U$-statistic of
degree $m$ with symmetric, mean zero kernel $\Phi$ that is applied to data of sample
size $n$. Thus, the essence of inequality \fer{hof} is to give a quantitative estimate
of the reduction of the variance of a $U$-statistic when the sample size $n$
increases. It is remarkable that, as soon as $m>1$ the functions $\Phi\left( X_{\bf
s}\right)$ are no longer independent. Nevertheless, the variance of the $U$-statistic
drops by a factor $m/n$.

With the essential help of Theorem \ref{kk} we prove

\begin{thm}\label{main} Let $T_n$ denote the sum \fer{stab}, where the random variables $X_j$
are independent copies of a centered random variable $X$ with bounded relative
$\lambda$-Fisher information, $1<\lambda<2$. Then, for each $n >1$, the relative
$\lambda$-Fisher information of $T_n$ is decreasing in $n$, and the following bound
holds
 \be\label{mm}
I_\lambda\left(T_n \right)\le \left(\frac{n-1}n
\right)^{(2-\lambda)/\lambda}I_\lambda\left(T_{n-1} \right).
 \ee
\end{thm}

\begin{proof}

In what follows, for $n>1$,  $S_n = \sum_{j \in [n]} Y_j$ will denote the
(unnormalized) sum of the independent and identically distributed random variables
$Y_j$. Likewise, $S_n^{(k)} = \sum_{j \not= k }Y_j$ will denote the leave-one-out sum
leaving out $Y_k$. Since $S_n = S_n^{(k)} + Y_k$, for each $k \in [n]$ we can write
 \[
 x= E[Y_k |S_n=x] + E[S_n^{(k)} |S_n=x].
 \]
On the other hand,  since the $Y_j$ are independent and identically distributed, we
have the identity
 \[
E[S_n^{(k)} | S_n=x] = (n-1) E[Y_k | S_n=x],
 \]
which implies that we can write
 \[
x= \frac n{n-1}E[S_n^{(k)} |  S_n=x].
 \]
Hence, for the relative fractional score of $S_n$ we  conclude that,  if $\upsilon_n =
(n-1)/n$,
 \[
\tilde\rho_\lambda(S_n) = E\left[\tilde\rho_{\lambda,\upsilon_n}(S_{n-1}^{(k)}) |S_n
\right],
 \]
for all $k \in [n]$, and hence
 \[
\tilde\rho_\lambda(S_n) = E\left[\frac 1n \sum_{k
\in[n]}\tilde\rho_{\lambda,\upsilon_n}(S_{n-1}^{(k)}) \big| S_n  \right].
 \]
Proceeding as in Lemma \ref{l1}, namely by using the fact that the norm of the
fractional relative score is not less than that of its projection, we obtain
 \[
 I_\lambda(S_n) = E\left[\tilde\rho_\lambda^2(S_n) \right] \le E\left[ \left(
 \frac 1n \sum_{k
\in[n]}\tilde\rho_{\lambda,\upsilon_n}(S_{n-1}^{(k)})\right)^2\right].
 \]
To this point, Theorem \ref{kk} yields
 \[
E\left[ \left(
 \frac 1n \sum_{k
\in[n]}\tilde\rho_{\lambda,\upsilon_n}(S_{n-1}^{(k)})\right)^2\right] \le (n-1)
\sum_{k \in [n]} \frac 1{n^2} E\left[\tilde\rho_{\lambda,\upsilon_n}^2(S_{n-1}^{(k)})
\right] = \frac{n-1}n I_{\lambda,\upsilon_n}(S_{n-1}).
 \]
If we suppose that the right-hand side in the previous inequality is bounded, we
obtained, for $n >1$ the bound
 \be\label{bou}
I_\lambda(S_n) \le \frac{n-1}n I_{\lambda,\upsilon_n}(S_{n-1}).
 \ee
To end up, let us choose $Y_k = X_k/n^{1/\lambda}$. In this case, $S_n = T_n$, where
$T_n$ is the sum \fer{stab}. Moreover
 \[
 S_{n-1} = \frac{X_1+ X_2
+ \cdots + X_{n-1}}{n^{1/\lambda}} = \left( \frac {n-1}n\right)^{1/\lambda}T_{n-1}=
\upsilon_n ^{1/\lambda}T_{n-1}.
 \]
On the other hand, thanks to formula \fer{stab},
 \[
 I_{\lambda,\upsilon_n}(S_{n-1}) =  I_{\lambda,\upsilon_n}(\upsilon_n^{1/\lambda}T_{n-1}) \le
 \upsilon_n^{-2(1- 1/\lambda)}I_\lambda(T_{n-1}).
 \]
Substituting into \fer{bou} gives the result.
\end{proof}

\begin{rmk}
{\rm Surprisingly enough, at difference with the case of the standard central limit
theorem, where $\lambda = 2$ and the monotonicity result of the classical relative
Fisher information reads $I(S_n) \le I(S_{n-1})$, in the case of the central limit
theorem for stable laws, the monotonicity of the relative $\lambda$-Fisher information
also gives a rate of decay. Indeed, formula \fer{mm} of Theorem \ref{main} shows that,
for all $n>1$
 \be\label{giu}
I_\lambda(T_n) \le \left( \frac 1n \right)^{(2-\lambda)/\lambda} I_\lambda(X),
 \ee
namely convergence in relative $\lambda$-Fisher information sense at rate
$1/n^{(2-\lambda)/\lambda}$.}
\end{rmk}

\begin{rmk}
{\rm This result allows to enlighten, at the level of Fisher information,  a strong
difference between the classical central limit theorem and the central limit theorem
for stable laws. In the former case, the domain of attraction is very large and
contains all random variables with finite variance, while the attraction in terms of
relative Fisher information is, without additional assumptions, very low (only
monotonicity is guaranteed). In the latter, the domain of attraction is very
restricted and contains only random variables with distribution which has the same
tails at infinity of the L\'evy stable law. However, in this case the attraction in
terms of the relative fractional fisher information is very strong, and it is
inversely proportional to  the exponent $\lambda$ which characterizes the L\'evy
stable law.}
\end{rmk}

\section{The relative $\lambda$-entropy}

The previous results  show that theoretical information techniques can be fruitfully
employed to obtain a new approach to the central limit theorem for stable laws. In
particular, the new concept of relative fractional Fisher information seems nicely
adapted to the subject, since most of the classical properties can be easily extended
to this case. In particular, subadditivity of the classical Fisher information with
respect to weighted convolution is shown to hold also for the relative fractional
Fisher information.

However, up to now, Fisher information has been considered as a useful instrument to
obtain results for Shannon's entropy functional defined in \fer{Shan}. The main
finding in this context was in fact the proof of the monotonicity of Shannon's entropy
on the weighted sums in the central limit theorem. If stable laws are concerned, at
difference with Fisher information, it seems difficult to find an explicit expression
of the (non-local) corresponding of Shannon's entropy, say the relative fractional
entropy.

One possible way to obtain a consistent definition would be to establish between the
relative fractional entropy and the relative fractional Fisher information the same
link which connects Shannon entropy to Fisher information through the solution to the
heat equation \fer{heat}. In addition to de Bruijn relation, further connections
between entropy and Fisher information have been established by Barron \cite{Bar} and
Carlen and Soffer \cite{CS}. Given two random variables $X$ and $Y$ of densities
$f(x)$ and, respectively, $g(x)$, the relative entropy  $H(X|Y)$ of $X$ and $Y$ is
defined as
 \[
H(X|Y) = H(f|g)= \int_\R f(x) \log \frac{f(x)}{g(x)} \, dx.
 \]
If $X$ is a random variable with a density $f(x)$ and arbitrary finite variance, and
$f_t(x)= f(x,t)$ denotes the solution to the heat equation \fer{heat} such that $f(x,
t=0)= f(x)$  then it holds \cite{Bar,CS}
 \[
H\left(X|Z\right) = \int_0^\infty I(f_t|z_{1+t})\, dt,
 \]
where $Z$ denotes as usual the Gaussian density of variance $1$, and $I(f_t|z_{1+t})$
is the relative (to the Gaussian of variance $1+t$) Fisher information defined in
\fer{fish-r}. In analogous way, one can define the relative (to the stable law)
relative entropy. Let $X$ be a random variable with density $f(x)$, and let $f_t(x)=
f(x,t)$ denote the solution to the fractional diffusion equation \fer{frac1} such that
$f(x, t=0)= f(x)$. Then, it appears natural to define the fractional relative entropy
of order $\lambda$ as
 \be\label{new-ent}
H_\lambda(X) = \int_0^\infty I_{\lambda, 1+t}(f_t)\, dt,
 \ee
where $I_{\lambda, 1+t}(f_t)$ denotes the relative fractional Fisher information
defined as in \fer{fish-r3}.

The relative fractional entropy given by \fer{new-ent} is well-defined. In fact, if
the random variable $X$ has a density $f(x)$, the solution to the fractional diffusion
equation \fer{frac1} has the density $f(x,t)$ of the sum $X_t = X+ t^{1/\lambda}Z$. By
inequality \fer{scal}
 \[
I_{\lambda, 1+t}(f_t) = I_{\lambda, 1+t}(X_t)\le (1+t) ^{-2(1-1/\lambda}
I_\lambda\left( X_t(1+t)^{-1/\lambda} \right).
 \]
In addition, since
 \[
X_t(1+t)^{-1/\lambda} = \left(\frac 1{1+t}\right)^{1/\lambda} X + \left(\frac
t{1+t}\right)^{1/\lambda} X,
 \]
thanks to inequality \fer{ep} it holds
 \[
I_\lambda\left( X_t(1+t)^{-1/\lambda} \right) \le I_\lambda\left( X \right).
 \]
Finally,
 \[
I_{\lambda, 1+t}(f_t) \le (1+t) ^{-2(1-1/\lambda} I_\lambda\left( X \right),
 \]
so that, integrating both sides from $0$ to $\infty$, we obtain
 \be\label{bbb}
H_\lambda(X) \le \frac\lambda{2-\lambda}I_\lambda\left( X \right).
 \ee
Analogously to the classical case, inequality \fer{bbb} implies that the domain of
attraction of the relative $\lambda$-Fisher information is a subset of the domain of
attraction of the relative $\lambda$-entropy.

Despite its complicated structure, as it happens in the classical situation, all
inequalities satisfied by the relative fractional Fisher information also hold for the
relative fractional entropy. By formula \fer{new-ent} we can easily show that Theorem
\ref{main} also is valid for the fractional relative entropy of order $\lambda$. In
the classical case, however, Csiszar--Kullback inequality \cite{Csi, Kul} allows to
pass from convergence in relative entropy to convergence in $L^1(\R)$. It would be
interesting to show that a similar result still holds in the case of the relative
fractional entropy, but at present this remains an open question.

\section{Conclusions}
Starting  from the pioneering work of Linnik \cite{Lin}, the role of Fisher
information to obtain alternative proofs of the central limit theorem has been
enlightened by a number of papers (cf. \cite{ABBN1, ABBN2, Bar, BCG2, BCG3, CS, Joh,
JB, MB, TV} and the references therein). Only recently, Bobkov, Chistyakov  and
G\"otze  \cite{BCG3} used of the relative Fisher information to study convergence to
stable laws.  At difference with the previous existing literature,  we  studied the
role of Fisher-like functionals in the central limit theorem for stable laws by
resorting to the new concept of relative fractional Fisher information. This nonlocal
functional   relies on the consideration of a linear fractional score function. As the
linear score function of a random variable $X$ identifies  Gaussian variables as the
unique random variables for which the score is linear, L\'evy symmetric stable laws
are here identified as the unique random variables for which the fractional linear
score is linear. This analogy is pushed further to show that the relative fractional
Fisher information, defined as the variance of the relative score, satisfies almost
all properties of the classical relative Fisher information. While the fractional
Fisher information represents in our opinion a powerful instrument to study
convergence towards symmetric stable laws, the role of the analogous of Shannon's
entropy in this context at present remains obscure, and will deserve further
investigations.

\section{Appendix}

To clarify that the domain of attraction of the fractional relative Fisher information
constitute a notion that could be fruitfully used, we retain of paramount important to
prove that this domain is not an empty subset of the classical domain of attraction of
the stable law. To this aim, we will provide in this appendix an explicit example of a
density which belongs both to the domain of attraction of the stable law, and to the
domain of attraction of the relative fractional Fisher information.

 To start with, let us briefly recall some information about the domain of attraction of a stable law. More details can be found in the book
\cite{Ibra71} or, among others, in the papers \cite{BLM}, \cite{BLR}. A centered
distribution $F$ belongs to the domain of normal attraction of the $\lambda$-stable
law \fer{levy} with distribution function $L_\lambda(x)$ if and only if $F$ satisfies
$|x|^\lambda F(x)\to c$ as $x\to -\infty$ and $x^\lambda (1-F(x))\to c$ as $x\to
+\infty$ i.e.
\begin{equation}\label{cardis}
\begin{aligned}
&F(-x)=\frac{c}{|x|^\lambda}+S_1(-x) \ \ \ \ \ \ {\rm and }\ \ \ \ \
\ \ 1-F(x)=\frac{c}{x^\lambda}+S_2(x) \ \ \ \ \ \ (x>0)\\
&S_i(x)=o(|x|^{-\lambda})\ \ \  \ {\rm as}\ |x|\to +\infty, \ \ \
i=1,2\\
\end{aligned}
\end{equation}
where $c=\frac{\Gamma(\lambda)}{\pi}\sin\left(\frac{\pi\lambda}{2}\right)$.

If the distribution function $F$ belongs to the domain of normal attraction of the
$\lambda$-stable law, for any $\nu$ such that $0< \nu<\lambda$ \cite{Ibra71}
 \be\label{mome} \int_{\R}|x|^{\nu} \,dF(x)<+\infty.
  \ee
The behavior of $F$ in the physical space \fer{cardis} leads to a characterization of
the domain of normal attraction of  the $\lambda$-stable law \fer{levy} in terms of
Fourier transform. Indeed, if $\widehat f$ is the Fourier transform of the
distribution function $F$ satisfying \fer{cardis}, then
 \be\label{fff}
1- \widehat f(\xi) = (1 - R(\xi))|\xi|^\lambda ,
 \ee
where
 \[
R(\xi) \in L^\infty(\R), \quad {\rm and}\quad |R(\xi)|  = o\left(1\right) \quad \xi
\to 0.
 \]
Let us consider a probability density $f(x)$ that belongs to the domain of attraction
of $L_\lambda$, with $\lambda >1$. Then, we have enough regularity to reckon the
Fourier transform of
 \[
\Psi_\lambda(x) = \mD_\alpha f(x) + \frac x\lambda f(x),
 \]
and, thanks to \fer{fff} we obtain
 \[
\widehat\Psi_\lambda(x) = i\xi|\xi|^{\lambda-2} \widehat f(\xi) + i\frac 1\lambda
\frac{\partial \widehat f(\xi)}{\partial \xi} = -i\xi|\xi|^{2\lambda-2} + V(\xi).
 \]
It is known that the leading small $\xi$-behavior of the singular component of the
Fourier transform \fer{expl} will reflect an algebraic tail of decay of the
distribution function (cf. for example Wong \cite{Wong}). In our case, since
$\widehat\Psi_\lambda(\xi)$ contains the term $\xi|\xi|^{2\lambda -1}$,  $\Psi(x)$
should decay at infinity  as  $|x|^{-2\lambda +1}$.

The leading
example of a function which belongs to the domain of attraction of the
$\lambda$-stable law is the so-called {\it Linnik distribution} \cite{L, L2},
expressed in Fourier variable by
  \be\label{Max-f}
  \widehat p_\lambda(\xi) = \frac 1{1+ |\xi|^\lambda}.
  \ee
For all $0<\lambda \leq 2$,  the function  \fer{Max-f} is the
characteristic function of a symmetric probability distribution. In addition, when
$\lambda > 1$, $\widehat p_\lambda \in L^1(\R)$, which, by applying the
inversion formula, shows that $p_\lambda$ is a probability density function.

The main properties of Linnik's distributions can be extracted from its representation
as a mixture (cf. Kotz and Ostrovskii \cite{KO}). For any given pair of positive
constants $a$ and $b$, with $0 < a < b \le 2$ let $g(s, a, b)$ denote the probability
density
 \[
g(s, a, b) = \left( \frac b\pi \sin\frac{\pi a}b \right) \frac{s^{a -1}}{1+ s^{2a} +
2s^a \cos\frac{\pi a}b }, \quad 0 <s<\infty.
 \]
Then, the following equality holds \cite{KO}
 \be\label{mix}
 \widehat p_a(\xi) = \int_0^\infty \widehat p_b(\xi/s) g(s, a, b)\, ds,
 \ee
or, equivalently
 \[
 p_a(x) = \int_0^\infty p_b(sx) g(s, a, b)\, ds.
 \]
This representation allows us to generate Linnik distributions of different parameters
starting from a convenient base, typically from the Laplace distribution
(corresponding to $b =2$). In this case, since $\widehat p_2(\xi)= 1/(1+|\xi|^2)$
(alternatively $p_2(x) = e^{-|x|}/2$ in the physical space), for any $\lambda$ with
$1<\lambda <2$ we obtain the explicit representation
 \be\label{oo}
  \widehat p_\lambda(\xi) = \int_0^\infty \frac {s^2}{s^2+ |\xi|^2} \,\,g(s, \lambda, 2)\, ds,
 \ee
or, in the physical space
 \be\label{lin2}
 p_\lambda(x) = \int_0^\infty \frac s2 \, e^{-s|x|} g(s, \lambda, 2)\, ds.
  \ee
Owing to \fer{lin2} we obtain easily that, for $1<\lambda<2$,  Linnik's probability
density is a symmetric and bounded function, non-increasing and convex for $x >0$.
Moreover, since Linnik's distribution belongs to the domain of attraction of the
stable law of order $\lambda$, $p_\lambda(x)$ decays to zero like $|x|^{1+\lambda}$ as
$|x| \to \infty$. These properties insure that there exist positive constants
$A_\lambda$ and $B_\lambda$ such that
 \be\label{nn}
  p_\lambda^{-1}(x) \le A_\lambda +B_\lambda |x|^{1+\lambda}.
 \ee
In reason of \fer{nn},  we obtain the bound
 \begin{equations}\label{bbn}
&I_\lambda(p_\lambda) =  \int_\R \left( \frac{\mD_{\lambda-1} p_\lambda(x)}{p_\lambda(x)} +
\frac x\lambda \right)^2p_\lambda(x)\, dx = \\
&\int_\R \left( \mD_{\lambda-1} +
 \frac x\lambda\, p_\lambda(x)\right)^2p_\lambda^{-1}(x)\, dx
  \le \int_\R  g_\lambda^2(x)
  (A_\lambda +B_\lambda |x|^{1+\lambda})\, dx.
 \end{equations}
In \fer{bbn}  we defined
 \[
  g_\lambda(x) = \mD_{\lambda -1}p_\lambda(x) + \frac x\lambda p_\lambda(x).
 \]
Explicit computations give
 \be\label{expl}
\widehat g_\lambda(\xi)= \frac{i \xi |\xi|^{2\lambda -2}}{\left( 1+
|\xi|^\lambda\right)^2}.
 \ee
As discussed before, the algebraic tail of decay of a function is given by the leading
small $\xi$-behavior of the singular component of its Fourier transform \fer{expl}.
Hence, while the leading singular component in the Linnik distribution \fer{Max-f} is
$|\xi|^\lambda$, which induces condition \fer{mome}, by \fer{expl} the leading
singular component in $\widehat g_\lambda^2$  is $|\xi|^{4\lambda -2}$, which would
imply that the right-hand side of \fer{bbn} is bounded, in reason of the fact that
$4\lambda -2 > \lambda +1$ for $\lambda >1$. This suggests that the relative
fractional Fisher information of Linnik's distribution is bounded.

An explicit proof of this property follows owing to representation \fer{mix}. To this extent, consider that the function $\widehat g_\lambda(\xi)$ defined in \fer{expl} can be also written in the form
 \[
 \widehat g_\lambda(\xi)= -i \frac{|\xi|^\lambda}\lambda\frac {d p_\lambda(\xi)}{d\xi}.
 \]
Therefore, differentiating under the integral sign in \fer{oo} we obtain the equivalent expression
 \be\label{opp}
\widehat g_\lambda(\xi)= \int_0^\infty \widehat h_\lambda(\xi/s) s^{\lambda-1} \,\,g(s, \lambda, 2)\, ds,
 \ee
where
 \be\label{qqq}
\widehat h_\lambda(\xi) = \frac{2i}\lambda \frac{\xi|\xi|^\lambda}{(1+ |\xi|^2)^2}.
 \ee
Since $1<\lambda <2$, it follows that both $\widehat h_\lambda$ and $\widehat
h_\lambda''$ belong to $L^2(\R)$. Hence,  by Plancherel's identity  we obtain
 \be\label{momo}
 \int_\R |x|^4 h_\lambda^2(x) \, dx =  \int_\R |x^2 h_\lambda(x)|^2 \, dx =
 \frac 1{2\pi}\int_\R |\widehat h_\lambda''(\xi)|^2 \, d\xi  < +\infty.
 \ee
Thus, for any given $R>0$
 \[
 \int_\R |x|^{1+\lambda} h_\lambda^2(x) \, dx \le \int_{\{|x| \le R\}} |x|^{1+\lambda} h_\lambda^2(x)\, dx + \frac 1{R^{3-\lambda}} \int_{\{|x| > R\}} |x|^{4} h_\lambda^2(x)\, dx \le
 \]
 \[
 \le (2R)^{1+\lambda}\int_\R h_\lambda^2(x) \, dx + \frac 1{R^{3-\lambda}} \int_{\R} |x|^{4} h_\lambda^2(x)\, dx.
 \]
Optimizing over $R$ we obtain
 \[
\int_\R |x|^{1+\lambda} h_\lambda^2(x) \, dx \le C_\lambda \left(\int_\R h_\lambda^2(x) \, dx\right)^{(3-\lambda)/4} \left(\int_\R |x|^4 h_\lambda^2(x) \, dx\right)^{(1+\lambda)/4},
 \]
where $C_\lambda $ is an explicitly computable constant. Note that, in view of \fer{momo}, the right-hand side of the previous inequality is bounded. Finally, by Jensen's inequality  we have
  \[
 \int_\R |x|^{1+\lambda} g_\lambda^2(x) \, dx  = \int_\R |x|^{1+\lambda}\left[ \int_0^\infty s h_\lambda(sx)\, s^{\lambda-1} \,\,g(s, \lambda, 2)\, ds \right]^2\, dx \le
  \]
 \[
\int_\R |x|^{1+\lambda} \left[ \int_0^\infty \left(h_\lambda(sx)\, s^{\lambda}\right)^2 \,\,g(s, \lambda, 2)\, ds \right]\, dx =
 \]
 \[
 \int_\R |x|^{1+\lambda}h_\lambda^2(x) \, dx \int_0^\infty s^{\lambda-2}\,g(s, \lambda, 2)\, ds. \]
By definition, the probability density $g(s, \lambda, 2) \in L^1(\R)\bigcap L^\infty(\R)$. This implies that, for $1 < \lambda <2$
 \[
\int_0^\infty s^{\lambda-2}\,g(s, \lambda, 2)\, ds < +\infty.
 \]
Hence the relative fractional Fisher information of the Linnik's density is bounded.

\bigskip

\bigskip \noindent

\noindent{\bf Acknowledgments:} This work has been written within the activities of
the National Group of Mathematical Physics of INDAM (National Institute of High
Mathematics). The support of the  project ``Optimal mass transportation, geometrical
and functional inequalities with applications'', financed by the Minister of
University and Research, is kindly acknowledged.





\end{document}